\documentclass[a4paper]{article}

\usepackage[colorlinks=true,breaklinks=true,bookmarks=true,urlcolor=blue,
     citecolor=blue,linkcolor=blue,bookmarksopen=false,draft=false]{hyperref}
\usepackage{amsmath,amssymb,amsfonts}
\usepackage{amsthm}
\usepackage{mathtools}
\usepackage{algorithmic}
\usepackage{float} 
\usepackage{multirow}
\usepackage{graphicx}
\usepackage{comment}
\usepackage{textcomp}
\usepackage{xcolor}
\usepackage[english]{babel}
\usepackage[utf8x]{inputenc}
\usepackage[T1]{fontenc}
\usepackage{graphicx}
\usepackage{delarray}
\usepackage{tikz}
\usepackage{pgfplots}
\usepackage{array}
\usepackage{framed,color}
\usepackage{scalefnt}
\usepackage{subfigure}
\usepackage{xspace}
\usepackage{graphics} 
\usepackage{epsfig}
\usepackage{times} 
\usepackage{hyperref}
\newtheorem{theorem}{Theorem}

\newtheorem{proposition}[theorem]{Proposition}
\newtheorem{corollary}[theorem]{Corollary}

\newtheorem{definition}[theorem]{Definition}

\title{Reputation-Based Information Design for Inducing Prosocial Behavior}

\begin{document}





\author{Alexandre Reiffers-Masson\thanks{Robert Bosch Centre for Cyber-Physical Systems, Indian Institute of Science, Bengaluru, 560012, India, \textit{email: reiffers.alexandre@gmail.com}}, Rajesh Sundaresan \thanks{Department of Electrical Communication Engineering,  Indian Institute of Science, Bengaluru, 560012, India, \textit{email: rajeshs@iisc.ac.in}} \footnotemark[1]}

\maketitle

\begin{abstract} 
We study the idea of information design for inducing prosocial behavior. As an example, we ground our study in the context of electricity consumption. Electricity utility providers would like to reduce the total power consumed by their residential consumer segment. Supply to this segment is often subsidized, and the saved power can be diverted to more profitable segments. Alternatively, the provider may be keen on earning carbon credits by inducing reduced consumption in this segment. The societal network in which the consumers reside may have a prevailing norm, for example, saving power is environment friendly and is considered good. Those that consume less are considered more prosocial and may derive a larger reputational benefit. How can the service provider, who is familiar with the prevailing norm and the consumption of all users, design suitable feedback signals that exploit reputation benefits to reduce net consumption? We call this a problem of information design and address this question in this paper. We consider a continuum of agents. Each agent has a different intrinsic motivation to reduce her power consumption. Each agent models the power consumption of the others via a distribution. Using this distribution, agents will anticipate their reputational benefit and choose a power consumption by trading off their own intrinsic motivation to do a prosocial action, the cost of this prosocial action and their reputation. We assume that the service provider can provide two types of quantized feedbacks of the power consumption. We study their advantages and disadvantages. For each feedback, we characterize the corresponding mean field equilibrium, using a fixed point equation. Besides computing the mean field equilibrium, we highlight the need for a systematic study of information design, by showing that revealing less information to the society can lead to more prosociality. In the last part of the paper, we introduce the notion of privacy and provide a new quantized feedback, more flexible than the previous ones, that respects agents' privacy concern but at the same time improves prosociality. The results of this study are not restricted to the framework of energy efficiency but are also applicable to congestion problems in road traffic and other resource sharing problems.
\end{abstract}

%


%
%
%

\section{Introduction.}
We ground our study of reputation-based information design for inducing prosocial behavior in the context of electricity consumption. Electricity utility providers are interested in the reduction of the power consumed by their residential consumer segment. Supply to this segment is often subsidized and if the electricity utility provider is able to save power, it will be possible to redirect it to more profitable segments. To do so, a classical economical approach would be to implement taxes.  However, in the residential consumer segment, these price-based policies can be difficult to implement for political reasons, they are often unpopular, and for engineering reasons (for instance, lack of knowledge of demand elasticities for energy-efficient durable goods) .
 
Recently, considerable attention has been paid to non-price interventions and policies that {\em nudge} consumers to conserve energy. In general, these methods are decentralized ways to change energy consumption behavior, are inexpensive to implement, and can be applied in any kind of society, whether developed or developing. One example in the context of demand response is the creation of a lottery with the distribution of energy coupons \cite{li2015energy,xia2017energycoupon,li2018mean}. Another example is the use of social comparison and reputation to improve the prosocial behavior of an agent (which is, in this case, the reduction of power consumption) \cite{schultz2007constructive,schultz2018constructive,allcott2011social}.   Other experiments have been performed and are summarized in \cite{allcott2010behavior}. The focus of our paper is to derive a game theoretical analysis of policies that use social comparison as a tool for improving prosocial behavior. This study is not the first of its kind and we highlight the difference with other related works in section \ref{sec:related works}. Social comparison has been proved to be a powerful means to induce prosocial behavior. In the context of the reduction of the total power consumed by the residential consumer segment, we suggest using social comparison instead of monetary reward because reduced consumption can be portrayed as a prosocial action. There is then a reputation benefit that can be attached to the reduction effort. Contrary to what one might expect at first glance, monetary rewards may reduce the expected prosocial action in some setting, as it has been demonstrated in \cite{benabou2006incentives}.

Our game theoretical model and analysis can be described as follows. The intensity of a person's prosocial orientation will be determined by three factors: her intrinsic interest in doing a prosocial action, the cost associated with this action and finally the desire to be {\em perceived} as generous and altruistic. The last factor can be seen as the agent's desire to have a good reputation in the society. By society we mean a group of people upon whom the agents want to make a good impression. From this perspective, the society can be even the individual herself, when the individual is viewing and assessing his own action, in which case we talk about self-signaling and identity. As observed in \cite{benabou2006incentives}, from a modeling perspective, there is no major difference between the society being a set of individuals or the agent herself.

The society assesses an individual's reputation based on its collective estimate of the agent's prosociality. This estimation is done using the knowledge of the actions of the different agents. This knowledge of actions can be total or partial. For example, in the context of household energy consumption, we can imagine two types of observation of the consumption of different households.

\begin{itemize}
\item In the first scenario, the electricity utility provider may publicize a full report to the society on the consumption of each agent\footnote{While this might seem an overkill in the electricity consumption context, such things are possible in the context of donations, voluntarism, committee work that help run an institution, etc. Our study should be seen in this broader context.}. The society is then able to estimate precisely how much more altruistic is one person's action than another person's action (thus ordering each person according to her respective altruism level).
\item In the second scenario, perhaps to address privacy concerns, the report may only show that the consumption of an agent is larger or smaller than a threshold. Others in the society, then, are only able to say if a person belongs to the more generous group or to the less generous one.
\end{itemize}

These two examples have the merit of demonstrating that it is possible, by manipulating the amount of information observed by members of the society, to mobilize societal opinion concerning an agent. We also notice that the second scenario respects agents' privacy concerns to a greater extent than the first one. Indeed, in the first scenario the action of each agent is fully reported to the society (there is a complete disclosure of the level of action). On the contrary, in the second scenario, the society gets to know only whether an agent belongs to the more generous group or to the less generous one.

Our goal is to show the importance of information design to improve prosocial behavior. The initial model is inspired from \cite{benabou2006incentives} which we extend in this paper to more complex information designs. We compare the impact of the different parameters on the efficiency and highlight the impact of the different information designs on prosocial behavior. From a mathematical point of view, extensions of the model \cite{benabou2006incentives} requires new proofs of existence, uniqueness under some conditions, and characterization of the mean field equilibria which we provide in this paper. Additionally, we discuss several possible extensions and future directions which we feel will be of use to researchers in the field.

\subsection{Organization and Main Results.}
The remainder of the paper is organized as follows. Section \ref{sec: model} introduces the game theoretical model. We describe the main components of the societal network which in our example is the complete (fully connected) network of all the agents, the agents' intrinsic motivations, their actions, the cost function associated with performing an action, and finally the reputational benefit associated to a given action. In this section, we also explain how the computation of reputation benefit can be understood as a signal extraction problem. We also introduce the different feedbacks of an agent's action. Section \ref{sec: equilibrium study} is the main section of this paper. We prove the existence and study the properties of mean field equilibria for the two feedback mechanisms. We are able to prove the existence of a mean field equilibrium and characterize it using a fixed point equation. For the second feedback, sufficient conditions for some uniqueness properties of the equilibrium follow naturally. We also compute explicitly the set of mean field equilibria for the feedbacks considered, under the assumption that the intrinsic values follows the uniform distribution. In section \ref{sec: shaping of feedback}, we will show how important the choice of feedback and the choice of the partition (into less prosocial and more prosocial groups) can be in improving the expected prosocial action. In section \ref{sec: privacy}, we formalize a notion of privacy, and then introduce a  new feedback
that improves privacy at the expense of the aggregate level of prosociality. In section \ref{sec:related works}, we discuss the related works and highlight the major difference of our work with the unified framework suggested in \cite{bergemann2016information,bergemann2017information} and with the norm based approach described in \cite{benabou2006incentives,benabou2011laws,ali2016image}. In section \ref{sec: extension}, we discuss many  possible extensions of this work and future directions. Finally, section \ref{sec: conclusion} concludes the paper with a very brief summary of our work.

\section{Model.}\label{sec: model}

\begin{table}[t]\caption{Main notations used throughout this paper}
\centering
\begin{tabular}{|p{0.20\columnwidth}|p{0.7\columnwidth}|}
\hline
{\it Symbol} & {\it Meaning}\\
\hline
$\mathcal{I}\coloneqq[0,1]$& Agent set. $i\in\mathcal{I}$ is the index of a given agent. \\
$\mathcal{A}\coloneqq[0,+\infty)$& Action space. $a_i\in\mathcal{A}$ is the action of agent $i$.\\
$w_i\in\mathbb{R}_+$& Intrinsic value. The empirical cdf of $(w_i)_{i \in I}$ is $F$.\\
$L(a_i)$& Quantized version of agent $i$'s action provided by the service provider to the societal network.\\
$C(\cdot)$ &  Cost function such that $C:[0,+\infty)\rightarrow \mathbb{R}_+$. \\
\hline
\end{tabular}
\label{tab:notation}
\end{table}

\vskip.2cm

\begin{table*}[t!]
  \centering
    \caption{Different $L$ functions with the associated utilities}
  \begin{tabular}{|l|c|c|}
  \hline
  
  {\it Feedback}& $L(a_i)$ & {\it Agents' objectives}\\
  \hline
    & & \\
    Type-A & $a_i$ & $ \displaystyle \max_{a_i\in[0,+\infty)}\left\{a_iw_i-C(a_i)+\beta\mathbb{E}\left[w_i\mid a_i\right]\right\}$ \\
    & & \\
  \hline 
    & & \\
    Type-B & $1_{a_i \geq \theta}$
    & $\begin{array}{lll}\max\left\{ \displaystyle \max_{a_i\in[0,\theta]}a_iw_i-C(a_i)+\beta\mathbb{E}\left[w_i\mid a_i< \theta\right],\right.\\
    \displaystyle \left. \hspace*{.9cm}\max_{a_i\in[\theta,+\infty)}a_iw_i-C(a_i)+\beta\mathbb{E}\left[w_i\mid a_i\geq \theta\right]\right\}\end{array}$ \\
    & & \\
  \hline
  \end{tabular}
  \label{tab:1}
\end{table*}

We assume that the societal network consists of a continuum number of agents, as we shall make precise soon. Each agent chooses an action from a set which is taken to be totally ordered in terms of prosocial behavior. An agent's choice of an action is based on three components, the \textit{agent's intrinsic motivation} (for instance the prosocial orientation of the agent), a \textit{cost} associated with her action and a \textit{reputational benefit}. The reputational benefit captures the effect of judgments and reactions of other members of the society towards an agent. A social service provider (e.g. utility service provider or government) is interested in the maximization of the global level of prosocial actions. For each agent, the service provider can manipulate an agent's reputational benefit by designing the information other agents get about this agent's choice. We assume that the agent's intrinsic motivation is \textit{private information} known only to herself but the distribution of agents' intrinsic motivation is \textit{common knowledge}. Agents interact with each other through the information fed by the social service provider and the consequent reputational benefits they derive. Moreover, since the number of agents is infinite, the global level of prosocial actions in the societal network is the outcome of a \textit{mean field equilibrium}. The service provider's problem is thus an \textit{Information Design Problem}. 

 The three main components of  the societal network are the following. i) Agents and actions: What is the set of agents, what are the agents' action spaces, and how do others in the society interpret agents' actions? ii) Intrinsic value and cost function: What is the information known about the agents' intrinsic values and the cost functions?  iii) Reputational benefit: How is reputational benefit quantified? In the rest of this section, we develop a model of the societal network and address each of these questions. The main symbols used in this paper are summarized in Table \ref{tab:notation}.

\textbf{Agents and actions: } The agent set is the continuum $\mathcal{I}\coloneqq[0,1]$. Let $i\in\mathcal{I}$ be a given agent. We suppose that each agent has a continuum of possible actions. We denote the action space as $\mathcal{A}\coloneqq[0,+\infty)$. The action of agent $i$ is denoted by $a_i\in\mathcal{A}$. For all $(a_{i},a'_{i})\in \mathcal{A}^2$, if $a_{i}>a'_{i}$ then agent $i$ performs a greater prosocial action when she chooses $a_{i}$ over $a'_{i}$. For instance if $a_{i}$ captures the energy {\em savings} of agent $i$, then the savings of energy is greater under the greater prosocial action $a_{i}$ than under the lesser prosocial action $a'_{i}$.

\textbf{Intrinsic value and cost function:} Each agent $i$ is endowed with an intrinsic value $w_i$ and gets a reward of value $a_iw_i$ for an action $a_i$. Larger the $w_i$, greater the propensity of the agent towards a more prosocial action. The cumulative distribution function (cdf) of $w_i$ is $F$. For each agent $i$, performing the action $a_i$ costs $C(a_i)$, where $C:[0,+\infty)\rightarrow \mathbb{R}_+$ is a convex and increasing function. $F$ and $C$ are common knowledge to all the agents and the service provider.

\textbf{Reputational benefit:} The reputational benefit is described by the following steps.
 \begin{enumerate}
 \item For each $i\in\mathcal{I}$, the intrinsic motivation of agent $i$ is given by her prosocial propensity $w_i$. When the agent $i$ chooses an action $a_i$, she reveals some information about $w_i$ to the service provider.
 \item Given $a_i$, which is agent $i$'s action, the service provider reveals a quantized version of  $a_i$, denoted by $L(a_i)$, to all agents in the societal network. The function $L$ is common knowledge. Our goal is to understand the consequence of various choices of $L$ by the service provider. Since $L$ controls the amount of information about agent $i$'s action $a_i$, this is an information design problem. A simple example is the privacy friendly feedback scheme
$L(a_i)=1_{a_i \geq \theta}$, with $\theta\in\mathbb{R}_+$.
\item When an action $a_i$ is taken and $L(a_i)$ is revealed to all, as discussed above, some information about the agent's private $w_i$ is also revealed to the societal network yielding a reputational benefit $\beta\mathbb{E}\left[w_i\mid \{L(a_k)\}_{k\in\mathcal{I}}\right]$, with $\beta\in\mathbb{R}_+$.  The notion that the reputation of an agent is based on the societal network's opinion about her intrinsic value $w_i$, arising from the information that her action $a_i$ reveals, has already been used in \cite{corneo1997theory, benabou2006incentives, benabou2011laws}. 
\end{enumerate}

The above is a model in the context of consumption of energy in a psychological experiment in \cite{schultz2007constructive}, which we now describe. It encompasses the following steps.

\textit{Step 1:} For a given day the service provider measures the consumption of energy of each household in a neighborhood.

\textit{Step 2:} If the consumption of a given household is below $\theta$, the service provider puts a green flag in front of that house; otherwise, nothing is done. 

\textit{Step 3:} Each household observes the flag in front of every other house and estimates the intrinsic motivation of that other household. 

For each $i\in\mathcal{I}$, we assume that the utility of agent $i$  for action $a_i$ is the sum of the rewards arising from her propensity for prosocial behavior, the cost function and the reputational benefit. Agent $i$ is thus interested in maximizing:
\begin{equation}\label{eqn: utility}
\max_{a_i\in [0,+\infty)}  U(a_i,w_i;a_{-i},w_{-i})   \coloneqq a_iw_i-C(a_i) +\beta\mathbb{E}\left[w_i\mid \{L(a_k)\}_{k\in\mathcal{I}}\right],
\end{equation}

with $a_{-i}$ (resp. $w_{-i}$) being the actions (resp. the intrinsic motivations) of all the agents except agent $i$.

Let $G$ be the cdf of actions $a_i$, $i\in\mathcal{I}$. This results in a certain feedback profile $\{L(a_k)\}_{k\in\mathcal{I}}$.
The best response of agent $i$ to $\{L(a_k),\;k\in\mathcal{I}\}$, which is a function of $G$, is given by:
\begin{equation}
a_i^*(w_i;G)=\text{arg}\max_{a_i\in[0,+\infty)}U(a_i,w_i;a_{-i},w_{-i}).
\end{equation}
We will assume that $a_i^*(w_i;G)$ is uniquely defined. See examples later. For a given $G$, let $TG$ be the distribution of the induced best response actions.
A mean field equilibrium is defined as follows.
\begin{definition} The distribution $G^*$ is a mean field equilibrium if $TG^*=G^*$.
\end{definition}
We refer the technical reader to a mathematically rigorous reformulation in the Appendix \ref{appendix: Game Reformulation}. The equilibrium $G^*$ will naturally depend on the feedback signal $L$.

One objective of the service provider could be to optimize the aggregate prosocial action:
\begin{equation}\label{eq:w}
W\coloneqq\max_{L(\cdot)} \int_{\mathcal{A}} b\,dG^*(b).
\end{equation}
Let $W_j$ be the expected prosocial action when type-$j$ feedback is provided, with $j\in\{A,B\}$. Type-A feedback reveals the action of an agent to the whole society. This feedback does not preserve the privacy of an agent's action. Under type-B feedback, the society only knows whether an agent belongs to the more prosocial group or to the less prosocial one. Therefore this provides better privacy than type-A feedback. See Table  \ref{tab:1} for a summary the two feedbacks and the utilities.


We now make the following assumptions about the cost function and the feedback functions for illustration of our main ideas.\\

\noindent \textbf{Assumption A:}
\begin{enumerate}
\item The cost function $C(a)=\frac{1}{2}\alpha a^2$.
\item $L(\cdot)$ is one of the functions defined in Table \ref{tab:1}. In type-B, the service provider can additionally control one threshold parameter, denoted $\theta$.
\end{enumerate}

\noindent We shall discuss extensions to general convex costs in section \ref{sec: extension} and additional feedbacks in section \ref{sec: privacy}.

\section{Equilibria for type-A and type-B feedback schemes.} \label{sec: equilibrium study}
In this section, we characterize and study the properties of mean field equilibria for reputational benefit feedback  types A and B. For ease of notation, we write $\mathbb{E}\left[w_i\mid L_j( \cdot)\right]$ for $\mathbb{E}\left[w_i\mid \{L_j(a_k)\}_{k\in\mathcal{I}}\right]$, when the feedback is type-$j$, $j\in\{\text{A, B}\}$. We will also write $\mathbb{E}\left[w_i\mid L_j( \cdot), L_j(a_i)\right]$ to draw the reader's attention to the feedback $L(a_i)$.

\subsection{Type-A equilibrium.}

We begin with a characterization of type-A equilibrium which is also portrayed in figure \ref{fig:Type_A_Distribution}.

\begin{theorem} \label{prop:fixed point type c} Under assumption A, there is a unique mean field equilibrium for type-A feedback. The prosocial action of player $i$ is the unique solution to the following equation:
\begin{equation}
a_i=\frac{w_i}{\alpha}+\beta(1- e^{-\frac{a_i}{\beta} }).
\end{equation}
Moreover,
\begin{eqnarray}
\mathbb{E}[w_i\mid L_A(\cdot)]&=&w_i\\\label{eq:derivative reputation type_A}
\frac{\partial \mathbb{E}[w_i\mid L_A(\cdot)]}{\partial a_i}&=&\frac{\beta}{\alpha}(1- e^{-\frac{a_i}{\beta} }).
\end{eqnarray}
\end{theorem}


\begin{figure}[t!]
    \centering
    \subfigure[Distribution of $a^*(w_i;G)$ under the type-A feedback.]{\includegraphics[scale=0.5]{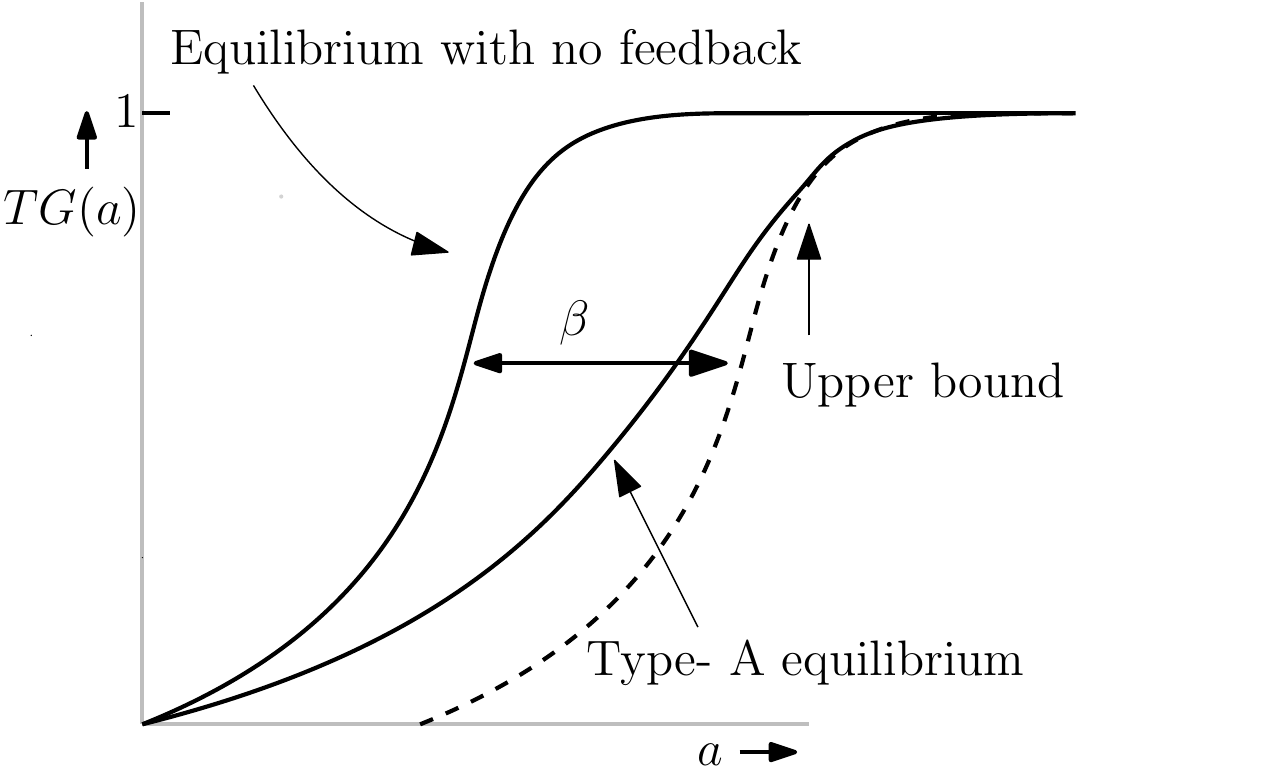}
\label{fig:Type_A_Distribution}}
    \subfigure[$a^*(w_i;G)$ under the type-A feedback.]{\includegraphics[scale = 0.5]{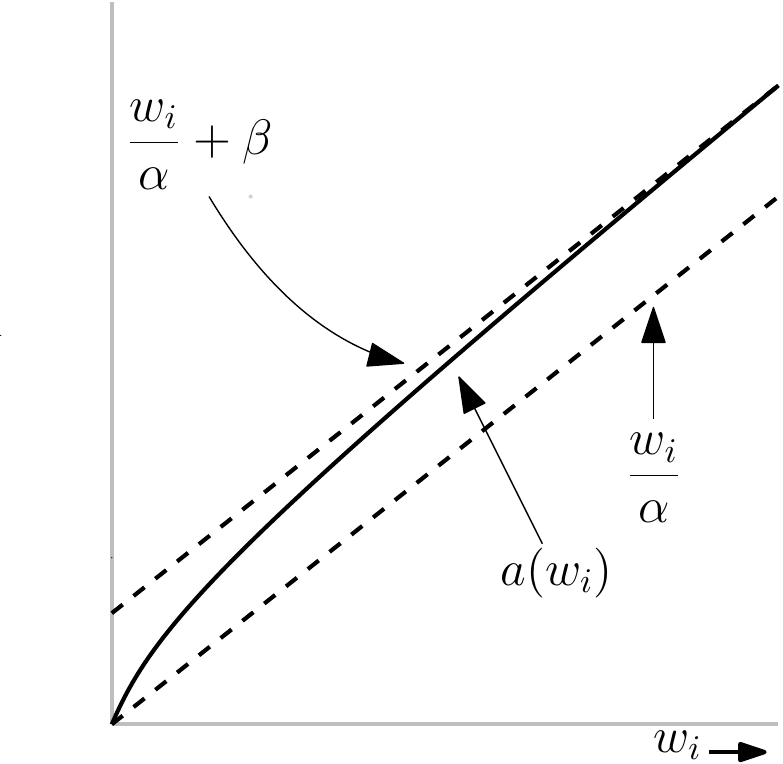}\label{fig:Type_A_Curve}}
    \caption{Distribution and shape of $a^*(w_i;G)$ under the type-A feedback.}
\end{figure}

{\bf Remark}: Theorem \ref{prop:fixed point type c} states that type-A feedback will result in an equilibrium that not only reveals $a_i$ to the whole society but also reveals the intrinsic value $w_i$ (since we will have $\mathbb{E}[w_i\mid L_A(\cdot)]=w_i$). In this case, the societal network learns everything about every agent. Also observe that feedback has pushed the agents towards more prosocial actions. In fact, $w_i/\alpha \leq a_i \leq w_i/\alpha + \beta$. In the limiting case $\lim_ {w_i\rightarrow +\infty}a_i=\frac{w_i}{\alpha}+\beta$ and $\lim_ {w_i\rightarrow 0_+}a_i=\frac{w_i}{\alpha}$. See figure \ref{fig:Type_A_Curve}.

{\bf Remark}: An examination of the proof below will indicate that it is nontrivial and nonstandard because it involves a functional fixed point equation. A similar exercise is carried out in B\'enabou and Tirole for the two-dimensional Gaussian case \cite{benabou2006incentives}. 
\begin{proof} We now provide the proof of Theorem \ref{prop:fixed point type c}. Since we will consider unilateral deviations, let us view $\mathbb{E}[w_i\mid L_A(\cdot)]$ as a function of $a_i$, while keeping all other actions fixed. Consider an agent with a specific $w_i$. For her  not to deviate, we must ensure the first order optimality condition which is
\begin{equation}\label{eq:foc}
w_i = \alpha a_i - \beta \frac{\partial \mathbb{E}[w_i\mid L_A(\cdot)]}{\partial a_i}.
\end{equation}
Since $G^*$ is an equilibrium action profile, \eqref{eq:foc} must hold for all $i\in\mathcal{I}$. Let us search for those strategies that make the right-hand side of \eqref{eq:foc} monotone increasing in $a_i$. By this assumption, whose validity we shall later check for our final solution, revealing $a_i$ is as good as revealing $w_i$ to all agents. This implies that at equilibrium $\mathbb{E}[w_i\mid L_A(\cdot)]=w_i$ for every player $i\in\mathcal{I}$ (i.e., every $w_i$). Plugging this into \eqref{eq:foc}, we have the following:

\begin{equation}
\mathbb{E}[w_i\mid L_A(\cdot)] = \alpha a_i - \beta \frac{\partial \mathbb{E}[w_i\mid L_A(\cdot)]}{\partial a_i}.
\end{equation}

Let us define $x(a_i)=\mathbb{E}[w_i\mid L_A(\cdot)]$. By taking the derivative in $a_i$ in the previous equation, we obtain:
\begin{equation}
\ddot x(a_i) = \frac{1}{\beta} (\alpha  - \dot x(a_i)).
\end{equation}
The solution to this linear differential equation is given by:
\begin{equation}\label{eq:lode}
\dot x(a_i)=\zeta e^{-\frac{a_i}{\beta} }+\alpha.
\end{equation}
For agent $i$ with $w_i=0$, we must have $a_i=0$. Indeed as stated previously the reputation of this agent is equal to $0$. So any nonzero action only adds to cost, and therefore $0$ is the best response. This boundary condition, using \eqref{eq:foc} and \eqref{eq:lode}, yields $\zeta=-\alpha$. By rearranging \eqref{eq:foc}, we have:
\begin{equation}\label{eq: a equi}
a_i=\frac{w_i}{\alpha}+\beta(1- e^{-\frac{a_i}{\beta} }).
\end{equation}
It is easy to see there is a unique solution $a_i$ to \eqref{eq: a equi} (intersection of a linear function and a concave increasing function that starts at a strictly positive value $w_i/\alpha$ but saturates at $w_i/\alpha+\beta$, see Figure \ref{fig:fixedpoint_A}). Finally, we check that the monotonicity assumption holds for the obtained $a_i$. From \eqref{eq:lode}, we have $\alpha a_i-\beta \dot x(a_i)=\alpha[a_i-\beta(1-e^{-\frac{a_i}{\beta}})]$. Its derivative is $\alpha(1-e^{-\frac{a_i}{\beta}})\geq 0$. This concludes the proof.
\begin{figure}[ht]
\centering
\includegraphics[scale=0.7]{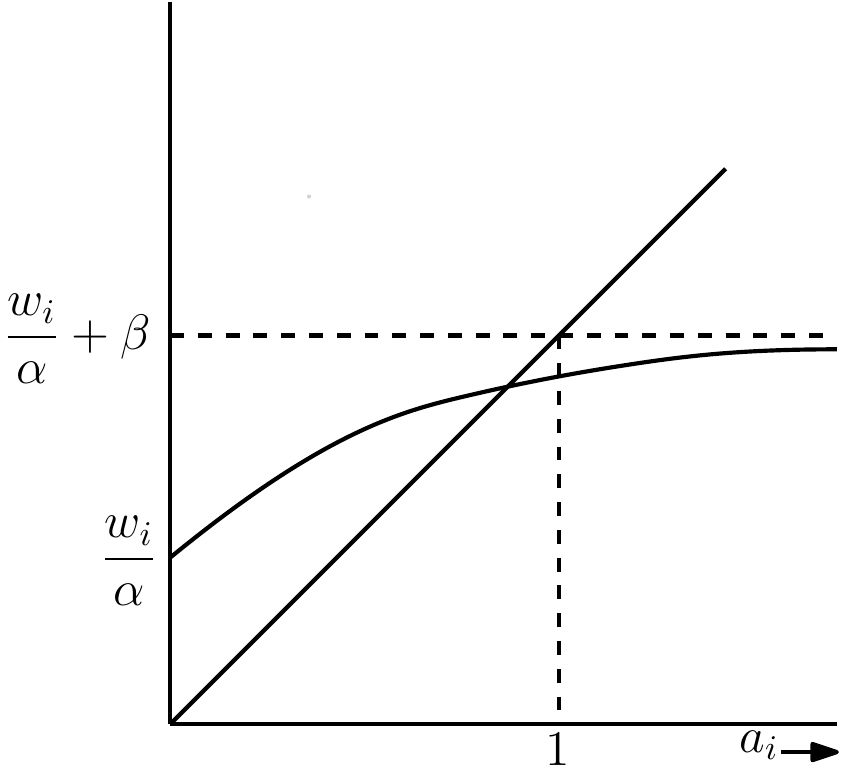}
\caption{Intersection between $a_i$ and $\frac{w_i}{\alpha}+\beta(1- e^{-\frac{a_i}{\beta} })$.}
\label{fig:fixedpoint_A}
\end{figure}

Moreover we can observe that because $\frac{w_i}{\alpha}-a_i+\beta(1- e^{-\frac{a_i}{\beta} })$ is decreasing in $a_i$ which satisfied the initial concavity assumption stated at the beginning of the proof.
\end{proof}

{\bf Remark}: It is clear from the proof that uniqueness crucially hinges on the boundary condition $a_i = 0$ when $w_i = 0$.

\subsection{Type-B equilibrium.}
As described in the previous section, agent $i$ computes her prosocial action given her intrinsic motivation $w_i$ and her cost function $C(\cdot)$. We state in this subsection, a fixed point equation that characterizes the best response $a_i$ to an action profile $G$ for each $w_i$ under the assumption that $L(\cdot)$ is of  type-B. Following this, we will demonstrate the existence of mean field equilibria.

For the type-B feedback function $L_B(\cdot)$, and with $G$ being the cdf of actions, agent $i$'s best response $a_i^{*}(w_i;G)$ is determined via:
 \begin{equation}\label{eq:utilitytypea}
\begin{array}{lll}
U_B(a_i^*(w_i;G),w_i)\coloneqq\\
\max\left\{\underbrace{a^*_{i1}(w_i;G)w_i-\displaystyle\alpha\frac{(a^*_{i1}(w_i;G))^2}{2}+\beta\mathbb{E}\left[w_i\mid L_B(\cdot),a^*_{i1}<\theta \right]}_{U_{1B}(a^*_{i1},w_i)},\right.\\
\hspace*{.94cm}\left.\underbrace{a^*_{i2}(w_i;G)w_i-\displaystyle\alpha\frac{(a^*_{i2}(w_i;G))^2}{2}+\beta\mathbb{E}\left[w_i\mid L_B(\cdot),a^*_{i2}\geq\theta\right]}_{U_{2B}(a^*_{i2},w_i)}\right\},\end{array}
\end{equation}
where (if the maxima below exist): 
\begin{eqnarray}
a^*_{i1}(w_i;G)&\coloneqq&\text{arg max}_{a_i\in[0,\theta)}U_{1B}(a_i,w_i),\\
a^*_{i2}(w_i;G)&\coloneqq&\text{arg max}_{a_i\in[\theta,+\infty)}U_{2B}(a_i,w_i),\\\nonumber
a_i^*(w_i;G)&\coloneqq&\displaystyle\left\{\begin{array}{ll}a^*_{i1}(w_i;G)&\text{if } U_{1B}(a^*_{i1}(w_i;G),w_i)>U_{2B}(a^*_{i2}(w_i;G),w_i),\\
a^*_{i2}(w_i;G)&\text{otherwise.}
\end{array}\right.
\end{eqnarray}

Here, under assumptions of existence, the candidate action level $a^*_{i1}(w_i;G)$ corresponds to the optimal action of agent $i$ if she were to perform an action below the threshold $\theta$. On the contrary, if agent $i$ were to perform an action $\theta$ or above , then the optimal choice would be $a^*_{i2}(w_i;G)$. The final choice $a_i^*(w_i;G)$ is the better of the two and thus the global optimum. In order to study the equilibria of this game, the first step is to derive an expression of $a^*_{i1}(w_i;G)$, $a^*_{i2}(w_i;G)$ and $a^*(w_i;G)$ as a function of $w_i$ by assuming that the reputational benefits are given. Then the second step will be to derive a closed form expression of reputational benefits. 

Let $c_1\coloneqq\mathbb{E}[w_i\mid L_B(\cdot),a_i<\theta]$ and $c_2\coloneqq\mathbb{E}[w_i\mid L_B(\cdot),a_i\geq\theta]$. Clearly $c_1$ is a function of $L_B(\cdot)$ and $1_{a_i<\theta}$ and hence is a constant for all $a_i<\theta$. Similarly $c_2$ is a constant for all $a_i\geq\theta$. 

\begin{proposition}\label{prop: best response type-B}(Best response to G) Under assumption A, if  $\Delta_B(G) \coloneqq c_2-c_1$ is positive,  then agent $i$'s best response to $G$ is
\begin{equation}\label{eq:a type b}
a^*_i(w_i;G)\coloneqq\left\{\begin{array}{ll}\displaystyle w_i/\alpha &\text{if } w_i\in[0,u) ~ \cup ~ [\alpha\theta,+\infty) \\
\theta&\text{if } u \geq w_i <\alpha \theta, \end{array}\right.
\end{equation}
with $u\in[0,\alpha\theta]$ satisfying:
\begin{equation}\label{eq: u best response}
u= \left[ \alpha\theta - \sqrt{2\alpha\beta\Delta_B(G)} \right]_+.
\end{equation}
\end{proposition}

\begin{proof}From measurability considerations, we have
\begin{eqnarray}\nonumber
c_1&\coloneqq&\mathbb{E}[w_i\mid L_B(\cdot),a_i<\theta]\text{ is independent of } a_i,\\\nonumber
c_2&\coloneqq&\mathbb{E}[w_i\mid L_B(\cdot),a_i\geq\theta]\text{ is independent of } a_i.
\end{eqnarray}
Case 1: Consider an individual with $w_i<\alpha\theta$. Her nonreputational utility function is depicted in Figure \ref{fig:test1}.
\begin{figure}[ht]
\centering
\includegraphics[scale=0.7]{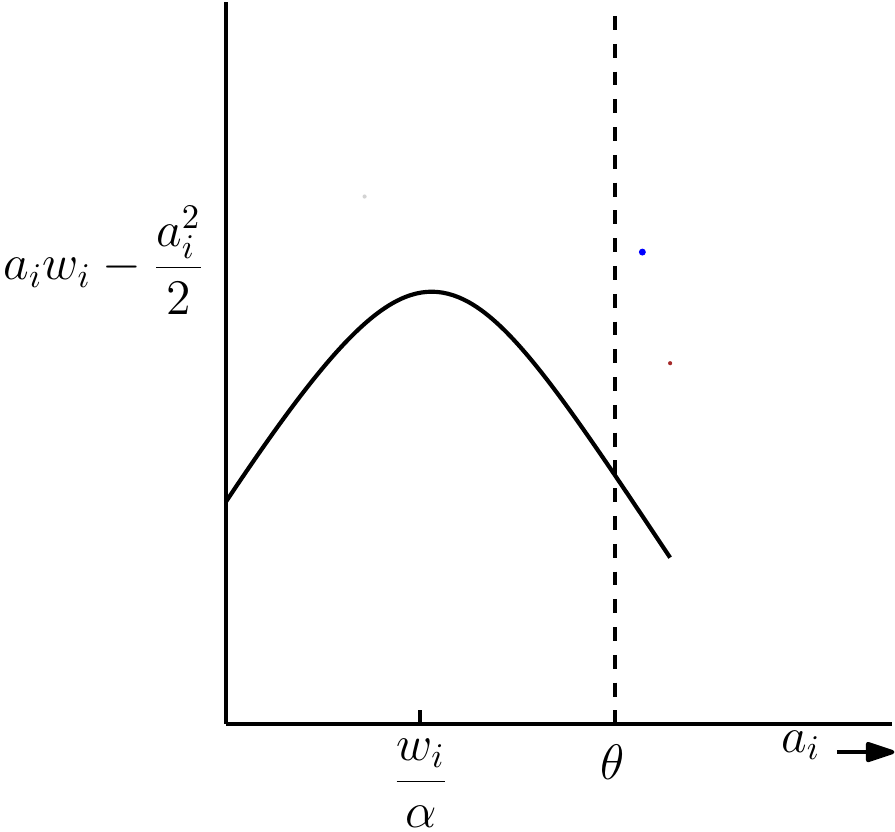}
\caption{Case 1: $\theta>\frac{w_i}{\alpha}$}
\label{fig:test1}
\end{figure}
Since $c_2$ is a constant independent of $a_i$ when $a_i\geq\theta$, it is optimal for the agent  not to play any $a_i$ other than $\theta$ in the set $[\theta,+\infty)$. Now let us consider $a_i<\theta$ versus $a_i=\theta$. When $a_i<\theta$, it is optimal for her to play $a_i=w_i/\alpha$, and she derives the utility:
\begin{equation}
U_B(a_{i1}^*(w_i;G),w_i)=\frac{w_i^2}{2\alpha}+\beta c_1.
\end{equation}
When $a_i=\theta$, she   derives the utility:
\begin{equation}
U_B(\theta,w_i)=\theta w_i-\frac{\alpha\theta^2}{2}+\beta c_2.
\end{equation}
Clearly then, those individuals with $w_i\in[0,\alpha\theta]$ such that 
\begin{eqnarray}\nonumber
\theta w_i-\frac{\alpha\theta^2}{2}+\beta c_2\geq\frac{w_i^2}{2\alpha}+\beta c_1&\Leftrightarrow & \frac{w_i^2}{2\alpha}-\theta w_i+\frac{\alpha\theta^2}{2}\leq \beta (c_2-c_1)\\\nonumber
&\Leftrightarrow& \displaystyle\frac{(w_i-\alpha \theta)^2}{2\alpha}\leq \beta(c_2-c_1)\\\nonumber
&\Leftrightarrow& \mid w_i-\alpha \theta \mid\leq \sqrt{2\alpha \beta (c_2-c_1)}\\\nonumber
&\Leftrightarrow & -\sqrt{2\alpha \beta (c_2-c_1)}\leq w_i-\alpha \theta\\\nonumber
&& \leq \sqrt{2\alpha \beta (c_2-c_1)},
\end{eqnarray}
will play $\theta$. Since we are considering $w_i<\alpha \theta$, agents with $w_i\in[u,\alpha\theta]$ where:
\begin{equation}\label{eq: u proof}
u=\left[\theta\alpha-\sqrt{2\alpha\beta c_2-c_1}\right]_+
\end{equation}
will play $\theta$. Others with $w_i<u$ will play $a_i=w_i/\alpha$. Hence,
\begin{equation}
a^*(w_i;G)\coloneqq\left\{\begin{array}{ll}\displaystyle\frac{w_i}{\alpha},&\text{if } w_i\in[0,u), \\
\theta,&\text{if } w_i\in[u,\alpha\theta).\end{array}\right.
\end{equation}
Case 2: Consider now an individual with $w_i\geq \alpha \theta$. Her nonreputational utility function is depicted in Figure \ref{fig:test2}.
\begin{figure}[ht]
\centering
\includegraphics[scale=0.7]{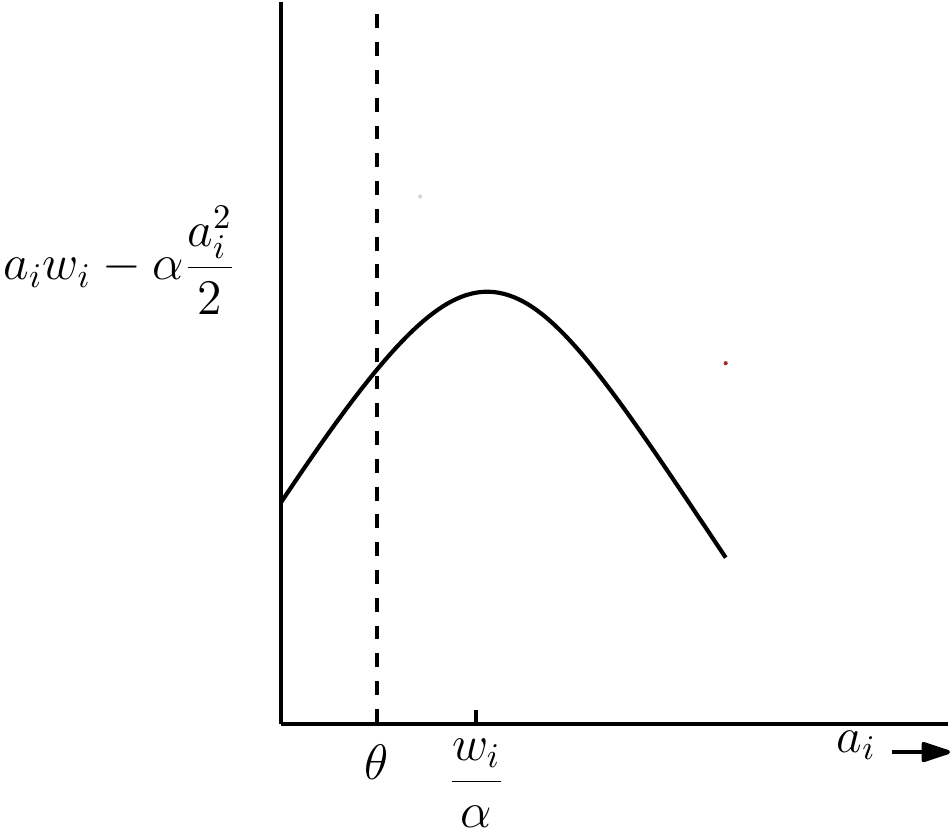}
\caption{Case 2: $\theta\leq\frac{w_i}{\alpha}$}
\label{fig:test2}
\end{figure}
Since $c_1$ is a constant independent of $a_i$ for $a_i<\theta$, and since $c_2>c_1$, it is optimal for her  not to play any $a_i$ smaller than $\theta$. It is also clear that if $a_i\geq\theta$, she   should pick $a_i=\frac{w_i}{\alpha}$.  Hence $a^*(w_i;G)=w_i/\alpha$ for all $i$ with $w_i\geq \alpha \theta$.

Finally, summarizing both cases, we get:
\begin{equation}
a^*(w_i;G)\coloneqq\left\{\begin{array}{ll}
\theta,&\text{if } w_i\in[u,\alpha\theta),\\
\displaystyle\frac{w_i}{\alpha},&\text{otherwise, } 
\end{array}\right.
\end{equation}
with $u$ as in \eqref{eq: u proof}. This proves Proposition \ref{prop: best response type-B}.

\end{proof}

The next corollary provides the distribution of the best response profile $a^*(w_i;G)$, and can be easily deduced from the previous Proposition \ref{prop: best response type-B}. 

\begin{corollary} Let $G$ be an action profile. Under assumptions A and $\Delta_B(G)\geq 0$, the best response profile $TG$ is:
\begin{equation}
TG(a)\coloneqq\left\{\begin{array}{lll}
F(\alpha a) &\text{if}& a\in[0,u/\alpha),\\
F(u) &\text{if}& a\in[u/\alpha,\theta),\\
F(\alpha a)&\text{if}& a\in[\theta,+\infty),
\end{array}\right.
\end{equation}
where $u$ is given in \eqref{eq: u best response} and depends on $G$ through $\Delta_B(G)$.
\end{corollary}

\begin{figure}[ht]
\centering
\includegraphics[scale=0.65]{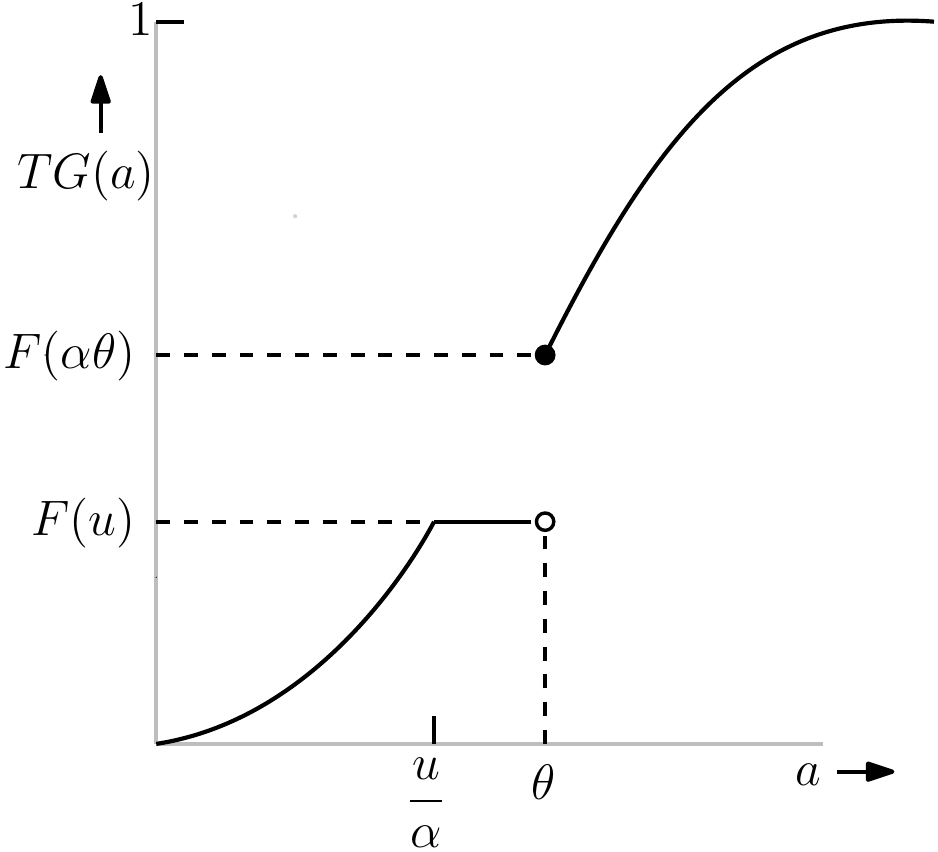}
\caption{Distribution of $a^*(w_i;G)$ under the type-B feedback. A segment of the population increase their prosocial action level to $\theta$.}
\label{fig:Type_B_Distribution}
\end{figure}

The distribution of the best response of the agents is depicted in Figure \ref{fig:Type_B_Distribution}.
 The shape $TG$ has the following properties. Firstly we note that if $w_i>\alpha \theta$, the prosocial action of agent $i$ is equal to $\frac{w_i}{\alpha}$ which is the same as in the case when the reputation benefit is not part of her utility ($\beta=0$). Also we can notice that by playing $\frac{w_i}{\alpha}$, she will reveal to the  service provider  her intrinsic motivation, since the  service provider  knows $\alpha$. In the case of $w_i<\alpha \theta$, two situations can occur. If $w_i<u$, then again the best response of agent $i$ will be to choose $\frac{w_i}{\alpha}$ and the same conclusion can be applied. On the other hand, if $w_i\in[u,\alpha \theta)$, then agent $i$ chooses $a_i^*(w_i;G)=\theta$ which implies that now her reputation is equal to $\mathbb{E}[w_i\mid L_B(\cdot),a_i\geq\theta]$. From this observation, we can deduce that, first, the agent makes the bare minimum effort to go into the next reputation category and, second, the agent does not fully reveal $w_i$ to the  service provider \footnote{Note that the service provider knows the actions, and from $a_i=\theta$ can infer that $w_i\in[\frac{u}{\alpha},\theta)$. However, the other agents only get to see $1_{a_i\geq \theta}$. The increased reputation benefit stems from this economical feedback to the other agents.}. This pattern of jump in prosocial actions for a segment of the population has been also observed in donations \cite{tirole2017economics}, when there are categories of donations. Indeed, the amount of donations are not distributed according to the uniform distribution but rather according to a  multimodal distribution where each mode corresponds to the minimum amount of donation needed to enter into a corresponding category. All others have either (1) sufficiently high intrinsic motivation ($w_i>\alpha\theta$) that playing lower without losing the reputational benefit only results in lower intrinsic benefit and so lower utility, or (2) sufficiently low intrinsic motivation ($w_i<u$) that playing higher will result in higher cost.
 
We assume that $F$ has a density $f(\cdot)$ and a finite expectation $\mathbb{E}[w_i]$. Define $Tc_1$ and $Tc_2$ as the $c_1$ and $c_2$ associated with $TG$ (See paragraph before Proposition \ref{prop: best response type-B}). These are then:
\begin{eqnarray} \label{eq:TC1}
Tc_1&\coloneqq&\frac{\int_0^uwf(w) \,dw}{F(u) }=u-\frac{\int_0^uF(w) \,dw}{F(u) }.\\\label{eq:TC2}
Tc_2&\coloneqq&\left\{\begin{array}{lll}\displaystyle\frac{\int_{u}^{+\infty} wf(w) \,dw}{1-F(u)}=\displaystyle u+\frac{\int_{u}^{+\infty} [1-F(w)] \,dw}{1-F(u)},&\text{if}& F(u)<1,\;\;\;\\
u &\text{if}& F(u)=1,
\end{array}\right.
\end{eqnarray}
where \eqref{eq:TC2} holds because the distribution $F$ is assumed to have finite mean. Thus,
 \begin{equation}
 \Delta_B(TG)=Tc_2-Tc_1 = 
 \left\{
 \begin{array}{lll}
 \displaystyle\frac{\int_{u}^{+\infty} [1-F(w)] \,dw}{1-F(u)}+\frac{\int_0^uF(w) \,dw}{F(u) },&\text{if} &F(u)<1\\
 u-\mathbb{E}[w_i],&\text{if}& F(u)=1.
 \end{array}\right.
\end{equation}
 At equilibrium, $TG^*=G^*$ and hence $\Delta_B(TG^*)=\Delta_B(G^*)$. Since $\Delta_B(TG)$ depends on only $F$ (known, fixed) and $u$ (to be determined), let us write:
\begin{equation}\label{eq:xb}
X_B(u)\coloneqq\left\{\begin{array}{ll}
\displaystyle\frac{\int_{u}^{+\infty} [1-F(w)] \,dw}{1-F(u)}+\frac{\int_0^uF(w) \,dw}{F(u) },&\text{ if }F(u)<1,\\
u-\mathbb{E}[w],&\text{ if } F(u)=1.
\end{array}\right.
\end{equation} 
Taking the derivative with respect to $u$ and simplifying, we get
\begin{equation}\label{eq:x'b}
 X_B'(u) = \left\{\begin{array}{lll}
\displaystyle f(u)\left[\frac{\int_{u}^{+\infty} [1-F(w)] \,dw}{(1-F(u))^2}-\displaystyle\frac{\int_0^uF(w) \,dw}{F(u)^2 }\right]&\text{if}&F(u)<1,\\
1,&\text{if}& F(u)=1.
\end{array}\right.
\end{equation}
From \eqref{eq:xb} and \eqref{eq:x'b}, we can conclude that $\lim_{u\downarrow 0} X_B(u)=\mathbb{E}[w]$, that $X_B(u)\geq u-\mathbb{E}[w]-\delta$ for a $\delta > 0$ and all $u$ sufficiently large and hence  $\lim_{u\uparrow +\infty} X_B(u)=+\infty$. Furthermore, we can also conclude that $X_B(u)$ is differentiable for all $u$ with $F(u)<1$, and further $X_B(u)$ continuous for all $0\leq u<+\infty$. From the relation in \eqref{eq: u best response} and the condition $\Delta_B(TG^*)=\Delta_B(G^*)$, we see that the intersection (or lack thereof) of the curves $X_B(u)$ in \eqref{eq:xb} and $(u - \alpha\theta)^2/(2 \alpha)$ will determine the equilibria. 

Two cases can occur. The first case is when there is no intersection between  $\beta X_B(u)$ and $(u - \alpha\theta)^2/(2 \alpha)$. In this case, $\beta X_B(u)$ is always greater, and from an equilibrium perspective, $u=0$ in \eqref{eq: u best response}, and all agents such that $w_i<\alpha\theta$ will play $\theta$. The rest of the agents will play $\frac{w_i}{\alpha}$. The second case is when $\beta X_B(u)$ and $(u - \alpha\theta)^2/(2 \alpha)$ intersect for some $u<\alpha \theta$, as in Figure \ref{fig: unique equilibrium type B}.
\begin{figure}[t]
\centering
\includegraphics[scale=0.65]{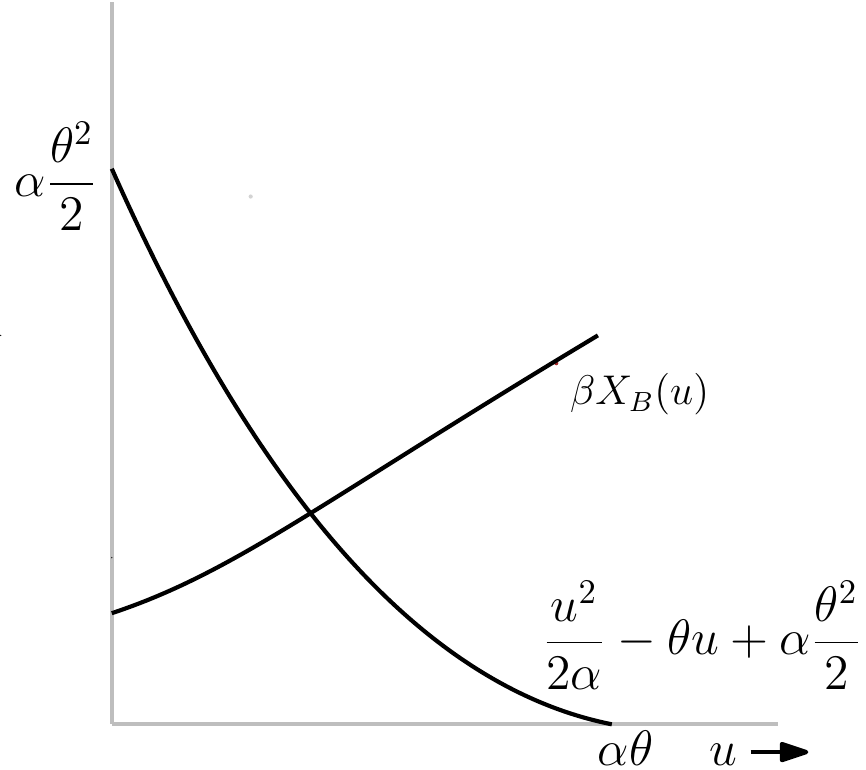}
 \caption{Unique equilibrium for type-B feedback}
 \label{fig: unique equilibrium type B}
\end{figure}
Clearly, if $F$ has a nontrivial density, then $X_B(\alpha\theta)>0$. (If $F(\alpha\theta)=1$, then $\mathbb{E}[w]<\alpha \theta$ since $F$ has a nontrivial density and $X_B(\alpha\theta)>0$ from the second case in \eqref{eq:xb}. If $F(\alpha\theta)<1$, then first formula in \eqref{eq:xb} applies and $X_B(\alpha \theta) \geq \left( \int_0^{\alpha\theta}F(w) \,dw \right) / F(\alpha\theta) >0$.) A sufficient condition for an interior $u^{MFE}$ is then $\theta>\sqrt{2\beta\mathbb{E}[w] / \alpha}$. We have thus established:

\begin{theorem} \label{cor:suffniquetypeB} A mean field equilibrium always exists. Moreover, the following hold:
\begin{enumerate}
\item If  $\beta X_B(u)$ and $(u - \alpha \theta)^2/(2\alpha)$ intersect at $u^{MFE}\in[0,\alpha\theta)$, then a mean field equilibrium exists with $a_i^*(w_i)$ as given in \eqref{eq:a type b} and $u^{MFE}$ solution of \eqref{eq: u best response}. 
\item If $\beta X_B(u)$ is greater than $(u - \alpha \theta)^2/(2\alpha)$ for all $u\in[0,\alpha\theta]$, then all agents such that $w_i<\alpha\theta$ will play $\theta$ and the rest of the agents will play $\frac{w_i}{\alpha}$.
\item If $\theta>\sqrt{\frac{2\beta\mathbb{E}[w]}{\alpha}}$, there is at least one intersection for the curves $\beta X_B(u)$ and $\frac{(u-\alpha \theta)^2}{2 \alpha}$.
\end{enumerate}
\end{theorem}

Uniqueness of the mean field equilibrium can be ensured if $\beta X_B(u)$ and $\frac{u^2}{2\alpha}-\theta u+\alpha \frac{\theta^2}{2}$ have a unique intersection. This is the case if $\beta X_B(u)$ is increasing in $u$.

\section{Shaping of feedback.}\label{sec: shaping of feedback}
We now consider two examples to highlight the need for a systematic study of information design. 
\subsection{Being economical with the truth can improve the net prosocial action.}
In type-A feedback the service provider revealed all the actions. In type-B feedback, the service provider revealed only the categories (more prosocial or less prosocial) to which an individual belonged. Let $W_A$ and $W_B$ denote the net prosocial actions of the population under type-A feedback and type-B feedback respectively. In the next proposition, we will give conditions such that $W_B$ is greater than $W_A$.
\begin{proposition}\label{proposition: W_B> W_A}
If $f(w)$ is decreasing in $w\in[0,+\infty)$ and if there exist $\theta$, $\alpha$ and $\beta$ such that the following conditions are satisfied:
\begin{eqnarray}\label{eq:condition 1 W_B >W_A}
\int_0^{\alpha \theta}\big(\theta-\frac{w}{\alpha}\big)f(w)\,dw&>& \beta,\\\label{eq:condition 2 W_B >W_A}
\beta \mathbb{E}[w] >\frac{\alpha \theta^2}{2},
\end{eqnarray}
then $W_B>W_A$.
\end{proposition}
\begin{proof}
If $f(w)$ is decreasing then $X_B(u)$ is increasing for $u\in(0,\alpha\theta)$ (see \cite{jewitt2004notes}). 
If $\beta\mathbb{E}[w]=\lim_{u\rightarrow 0}\beta X_B(u)>\frac{\alpha\theta^2}{2}$, then $u$ is always is equal to 0 according to Proposition \ref{prop: best response type-B}. Observe from Theorem \ref{cor:suffniquetypeB} that $W_A\leq \frac{\mathbb{E}[w]}{\alpha}+\beta$. Hence,
\begin{eqnarray}\nonumber
W_B-W_A&\geq&W_B-\frac{\mathbb{E}[w]}{\alpha}-\beta\\\nonumber
&\geq&\int_0^u\frac{w}{\alpha}f(w)\, dw+\theta\int_u^{\alpha\theta}f(w)\, dw+\int_{\alpha \theta}^{+\infty}\frac{w}{\alpha}f(w)\,dw-\frac{\mathbb{E}[w]}{\alpha}-\beta\\\nonumber
&=&\theta\int_u^{\alpha\theta}f(w)\,dw-\int_u^{\alpha\theta}\frac{w}{\alpha}f(w)\,dw-\beta>0.
\end{eqnarray}
where the last equality follows from \eqref{eq:condition 1 W_B >W_A}. This concludes the proof.
\end{proof}
We now provide an example where the conditions of Proposition \ref{proposition: W_B> W_A} are satisfied. 
In Figure \ref{fig:condition area}, we first understand why these two conditions are in opposition. Indeed, if we want \eqref{eq:condition 1 W_B >W_A} to be satisfied we will need a small $\beta$ so that the area of the crossed region is at least $\beta$. But if $\beta$ is too small we will reduce the area of the square dotted region which will affect \eqref{eq:condition 2 W_B >W_A}. Similar conclusions can be drawn for $\theta$. By increasing $\theta$, it will be easier to satisfy \eqref{eq:condition 1 W_B >W_A}, but more difficult to satisfy \eqref{eq:condition 2 W_B >W_A}. 

Let us consider that the density $f(w)=\frac{k}{\lambda}(\frac{w}{\lambda})^{k-1}e^{-(\frac{w}{\lambda})^k}$  (Weibull distribution) with a scale parameter ($\lambda$) equal to $1$ and shape parameter ($k$) equal to $0.5$. When the shape parameter is lower than $1$, $f(w)$ is decreasing in $w$. Take $\alpha=1$. The results of a numerical simulation in Figure \ref{fig:conditions}, indicate that the pairs $(\beta,\theta)$ in the darkened area satisfy the conditions of Proposition \ref{proposition: W_B> W_A}. For these parameters, being economical with the feedback improves the net prosocial action.  
\begin{figure}[ht]
    \centering
    \subfigure[Graphical illustration of \eqref{eq:condition 1 W_B >W_A} and \eqref{eq:condition 2 W_B >W_A} when $\alpha=1$.]{\includegraphics[width = 0.5\textwidth]{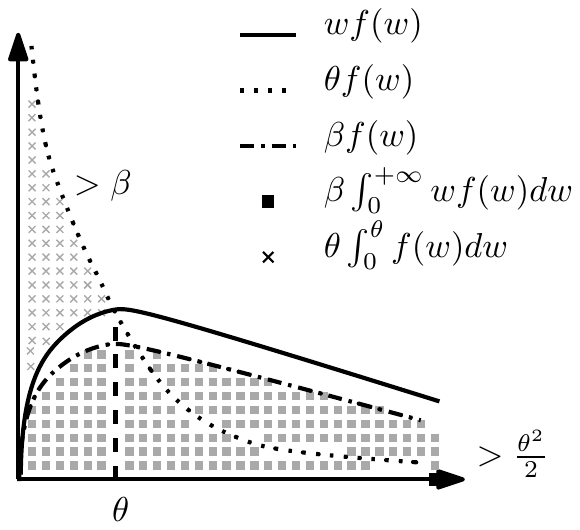}\label{fig:condition area}}
    \hspace{1cm}
    \subfigure[Parameters $\beta$ and $\theta$ such that the conditions of  Proposition \ref{proposition: W_B> W_A} are satisfied for a Weibull distribution.]{\includegraphics[width = 0.4\textwidth]{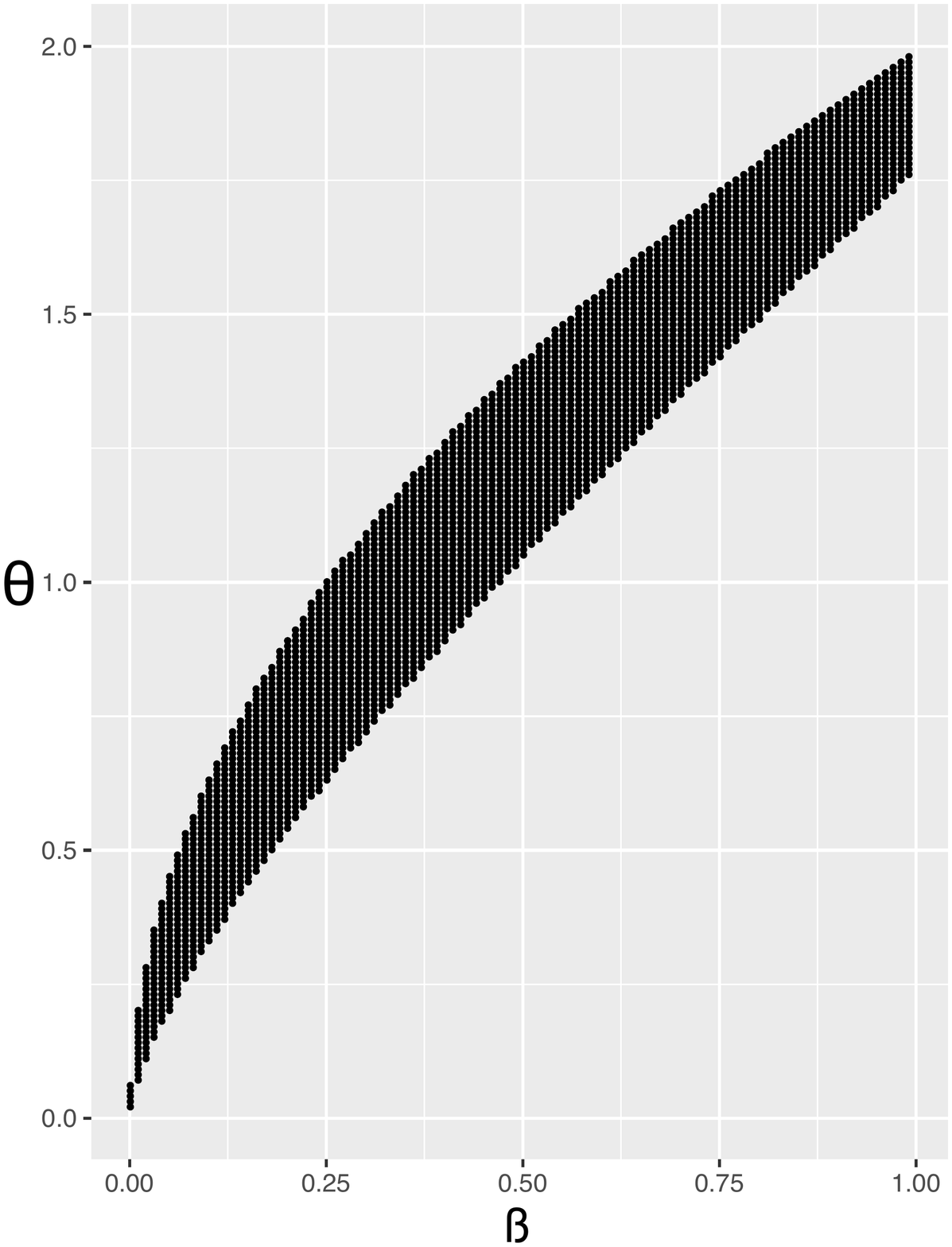}\label{fig:conditions}}
    \caption{Situations when $W_B > W_A$.}
\end{figure}



\subsection{The uniform distribution case and an explicit optimization.}


In this subsection, we study the example when $F$ has the uniform distribution over $[0,1]$. Under this assumption, an agent's intrinsic motivation is her rank in the society, i.e., $F(w_i) = w_i = i$. The utility of agent $i$ can be rewritten as follows:
\begin{equation}
U(a_i,i,a_{-i})=a_i i -C(a_i)+\beta\mathbb{E}[i\mid\{L(a_k)\}_{k\in\mathcal{I}}].
\end{equation} 

For simplicity, assume $\alpha \beta<1$. Since we will be interested in making individuals jump to more prosocial actions, from Figure \ref{fig:Type_B_Distribution}, we have that $\alpha \theta \leq 1$.

 It can be easily checked that $X_B(u) \equiv 0.5$. The mean field is determined by the intersection between a constant function $\beta X_B(u) \equiv \beta 0.5$ and the decreasing polynomial function $(u - \alpha \theta)^2/(2 \alpha)$ in the interval $u \in [0, \alpha\theta]$. 
 If there is no intersection then $u^{MFE}=0$. Thus $u^{MFE}=[\theta\alpha-\sqrt{\alpha\beta}]_+$. Furthermore, by Theorem \ref{cor:suffniquetypeB}:
\begin{equation}
a^*(i)\coloneqq\left\{\begin{array}{ll}\displaystyle {i/\alpha}&\text{if } i\in[0,[\alpha\theta-\sqrt{\alpha\beta}]_+) ~ \cup ~ [\alpha\theta,1] \\
\theta&\text{otherwise.}\end{array}\right.
\end{equation}
The net prosocial action under type-B feedback is given by:
\begin{eqnarray}\nonumber
 W_B&=& \displaystyle\frac{1}{2\alpha}[w^2]_0^{[\alpha\theta-\sqrt{\alpha\beta}]_+}+\theta(\alpha\theta-[\alpha\theta-\sqrt{\alpha\beta}]_+)+\frac{1}{2\alpha}[w^2]_{\alpha\theta}^1\\ \nonumber \\
&=&\left\{\begin{array}{lll}\displaystyle\frac{1}{2\alpha}+\frac{\alpha\theta^2}{2}&\text{if}&\theta\leq\sqrt{\beta/\alpha}\vspace{0.5 cm}
\\
\displaystyle \frac{\alpha\beta+1}{2\alpha}&\text{if}&\theta>\sqrt{\beta/\alpha}
\end{array}\right. \\\nonumber\\\nonumber\\
& =& \frac{1}{2\alpha} + \frac{1}{2} \min \left\{ \alpha \theta^2,  \beta \right\}.
\end{eqnarray}
Note that if $\theta\leq\sqrt{\beta/\alpha}$, then  $W_B$ is increasing in $\theta$ else $W_B$ is independent of $\theta$. Therefore the optimal $\theta$ to use will be any $\theta$ that satisfies $\frac{1}{\alpha} \geq \theta \geq \sqrt{\frac{\beta}{\alpha}}$. 

\section{A new feedback for trading off privacy for efficiency.}\label{sec: privacy} 
Our initial intuition was that type-A feedback will maximize the level of aggregate prosocial action. Surprisingly, as highlighted in Figure \ref{fig:conditions} for the Weibull distribution, this is not always the case. Additionally, type-A feedback may not satisfy the privacy concerns of agents. In this section, we will first define a privacy measure associated with a feedback. We then extend type-B feedback by allowing multiple thresholds. Finally, we propose an optimal information design problem.   
\subsection{A privacy measure.}

Our measure of privacy is the extent to which the society is uncertain about an agent's intrinsic motivation, averaged across the population. Recall that the society observes $L(a_i^*(w_i))$, therefore $\mathbb{E}[w_i\mid L(\cdot)]$ is the minimum mean squared error estimate of $w_i$. We define the privacy measure as the mean square error over the population:
\begin{equation}
V(L)=\int_0^{+\infty}(w_i-\mathbb{E}[w_i\mid L(\cdot)])^2 f(w_i)\,dw_i.
\end{equation}
When $V(L)=0$, the society is able to infer precisely the true value of $w_i$ for each agent $i$. The higher $V(L)$ the greater is our measure of privacy.

\subsection{Type-Bm feedback (multi-threshold type-B feedback).}

We explore the following feedback:
\begin{equation}
    L_{Bm}(x)= n-1 \mbox{ if } x\in[\theta_{n-1},\theta_n),
\end{equation}
with $0=\theta_0<\theta_1<\cdots<\theta_N=+\infty$. Let us consider the following candidate equilibrium, where the strategy of agent $i$, if $w_i\in[v_{n-1},v_n)$, is given by:
\begin{align}\label{eq:payofftype_bm}
\text{arg } \max_{a_i\in[\theta_{n-1},\theta_n) }\left\{a_iw_i-C(a_i)+\beta\mathbb{E}[w_i\mid L_{Bm}(\cdot),a_i\in [\theta_{n-1},\theta_n)]\right\},   
\end{align}
with $0=v_0<v_1<\cdots<v_{\max}$, where $v_{\max}$ is the smallest $v$ with $F(v)=1$.
Only a quantized signal of agent $i$'s action is revealed, whose $w_i\in[v_{n-1},v_n)$, is revealed, and therefore her reputational benefit comes from 
\begin{eqnarray}\nonumber
Y(v_n,v_{n-1})&\coloneqq &\mathbb{E}[w_i\mid L_{Bm}(\cdot),a_i\in [\theta_{n-1},\theta_n)]\\\nonumber
& =& \mathbb{E}[w_i\mid w_i\in[v_{n-1},v_n)]=\frac{\int_{v_{n-1}}^{v_n}wf(w)\,dw}{F(v_n)-F(v_{n-1})}.
\end{eqnarray}
Note that $Y(v_{n+1},v_{n})\geq Y(v_n,v_{n-1})$ for all $n\in\{1,\ldots,N-1\}$. 

An agent $i$, with $w_i\in[v_{n-1},v_n)$, who is evaluating a deviation from the candidate equilibrium, will face the following optimization problem:
\begin{equation}
    \max_{n\in\{1,\ldots,N\}}\left\{\Big[\frac{w_i}{\alpha}\Big]_{\theta_{n-1}}^{\theta_n}w_i-\frac{\alpha}{2}\Big(\Big[\frac{w_i}{2\alpha}\Big]_{\theta_{n-1}}^{\theta_n}\Big)^2+\beta Y(v_n,v_{n-1})\right\},
\end{equation}
with $[x]_a^b=\min\{\max\{x,a\},b\}$. This is obtained by first solving the optimization problem \eqref{eq:payofftype_bm} within the interval $a_i\in[\theta_{n-1},\theta_n)$, the solution to which is $\displaystyle\Big[\frac{w_i}{\alpha} \Big]_{\theta_{n-1}}^{\theta_n}$, followed by an optimization over $n\in\{1,2,\ldots,N\}$. 

If we assume that for all $n'\geq n^*$ which may depend on $i$, 
\begin{equation}\label{eq: condition existence Type-D}
\theta_{n'}\left(w_i-\frac{\alpha\theta_{n'}}{2}\right)+\beta Y(v_{n'+1},v_{n'})>\theta_{n'+1}\left(w_i-\frac{\alpha\theta_{n'+1}}{2}\right)+\beta Y(v_{n'+2},v_{n'+1}),
\end{equation}
with $n^*$ such that $\frac{w_{i}}{\alpha}\in[\theta_{n^*-1},\theta_{n^*})$
, then this optimization problem can be rewritten as:
\begin{equation}
\max \left\{\frac{w_i^2}{2\alpha}+\beta Y(v_{n^*},v_{n^*-1}),\quad\theta_{n^*}(w_i-\frac{\alpha\theta_{n^*}}{2})+\beta Y(v_{n^*+1},v_{n^*}) \right\}.
\end{equation}
Now the deviation of player $i$ can be restricted to a choice in $[\theta_{n^*-1},\theta_{n^*})$ (in fact $\frac{w_i}{\alpha}$) or a choice in $[\theta_{n^*},\theta_{n^*+1})$ (in fact $\theta_{n^*}$). The equilibrium condition is therefore equivalent to the case where player $i$ with $w_i=v_n$ is indifferent to choosing $\frac{v_n}{\alpha}$ or $\theta_{n}$. Therefore, we need that, for all $n\in\{1,\ldots,N-1\}$, $v_n$ satisfies the following vector fixed point equation:
\begin{equation}
\frac{v_n^2}{2\alpha}+\theta_n \left( \frac{\alpha\theta_n}{2}-v_n\right) = Y(v_{n+1},v_{n})-Y(v_{n},v_{n-1}),
\end{equation}
with $v_n\in[\alpha\theta_{n-1},\alpha\theta_n)$ for all $n\in\{1,\ldots,N-1\}$. 

Consider now an example where $F$ is a uniform distribution over the interval $[0,1]$. The previous fixed-point equation can be rewritten as a system of equations:
\begin{equation}\label{eq:system fixed points type bm}
\displaystyle\frac{(v_n-\alpha \theta_n)^2}{2\alpha}= \beta \frac{v_{n+1}-v_{n-1}}{2},
\end{equation}
i.e. $\alpha \theta_n-v_n=\sqrt{\alpha \beta (v_{n+1}-v_{n-1})}$ (because $v_n<\alpha \theta_n$), with $v_0=0$ and $v_N=1$.
\begin{figure}
    \centering
    \includegraphics[scale=0.9]{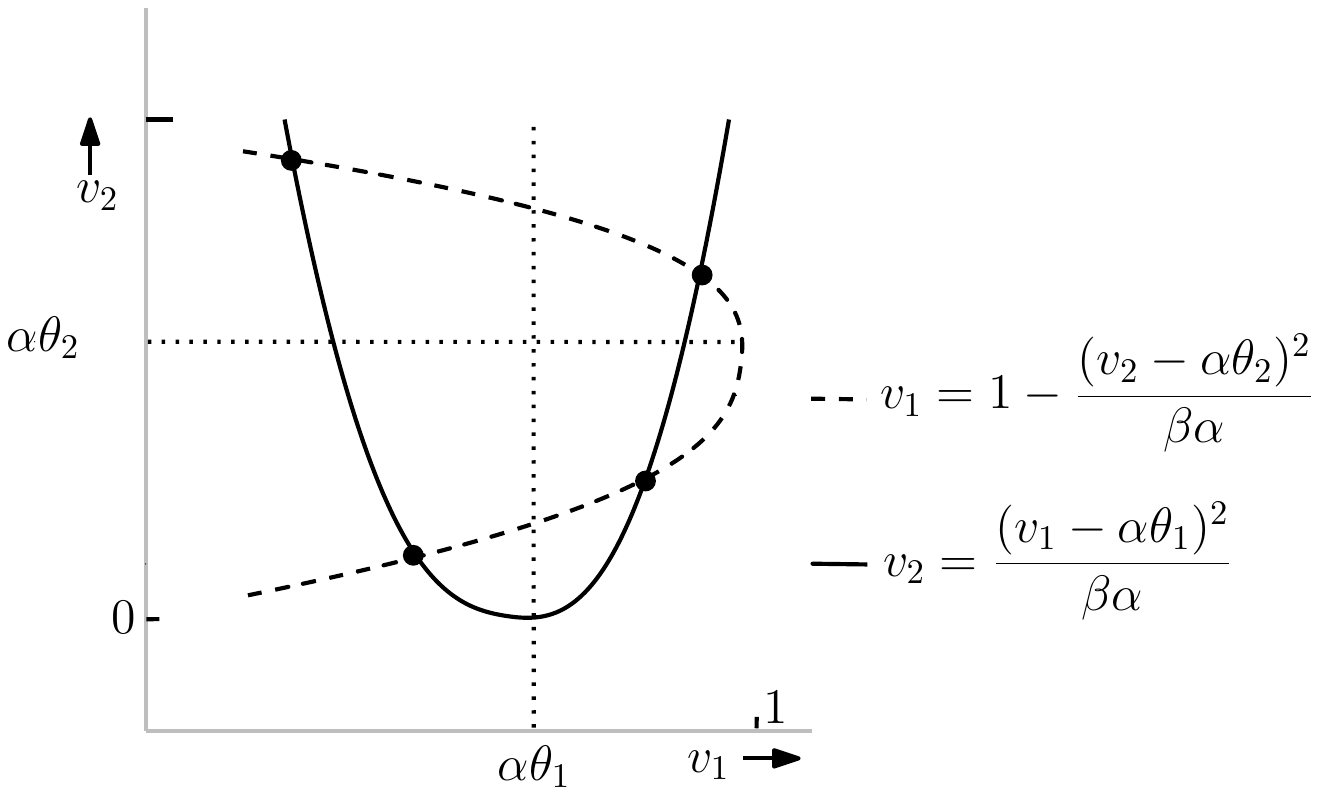}
    \caption{Fixed point equations for Type-Bm feedback, under the condition that $v_1\in(0,\alpha \theta_1)$ and $v_2\in(\alpha \theta_1,\alpha \theta_2)$.}
    \label{fig:fixed_point_type_BM}
\end{figure}
For two thresholds, the implicit plots associated to the fixed point systems \eqref{eq:system fixed points type bm} are depicted in Figure \ref{fig:fixed_point_type_BM}. We fix $\theta_1=\frac{1}{3}$, $\theta_2=\frac{2}{3}$, $\alpha=1$ and $\beta=0.1$. We obtain $v_1=0.14$ and $v_2=0.37$. With these parameters $W_{Bm}=0.562$. As a comparison, type-B feedback with $\theta=1/3$ yields a lower $W_{B}=0.55$. 

\subsection{Service provider optimization problem.}
The service provider may wish to guarantee a given level of privacy, by designing a feedback $L$ that will satisfy the following constraint:
\begin{equation}
V(L)\geq V_{\min}.  
\end{equation}
If $G^*_L$ is the equilibrium action profile associated with $L$, then the service provider should solve the following constrained optimization problem:
\begin{equation}
W\coloneqq\max_{L(\cdot)} \int_0^{a_{\max}}b\,dG^*(b).
\end{equation}
subject to:
\begin{equation}
V(L)\geq V_{\min}.  
\end{equation}
A solution to this problem is beyond the scope of this paper. The problem is posed only to direct the research community's attention to such optimizations with the following additional note in relation to practice. While solutions to such problems may be extremely useful from an understanding perspective, we should also keep in mind that the thresholds should be restricted to ``pivotal'' values that the members of the society can easily relate to. For instance, in the context of sponsorship levels for a event, suppose that the optimal two-level feedback for maximizing donor contributions is to set a threshold of 1023.11 dollars. It may however be more natural to set the threshold at the pivotal 1000 dollars in a ``digital society'' or the pivotal 1024 dollars in a society more comfortable with bits.

\section{Related works.} \label{sec:related works}

This section is dedicated to the related works and especially the main difference between our work with (1) the information design paradigm framed in \cite{bergemann2017information} and (2) the 
model suggested in \cite{benabou2006incentives}.

\textbf{Information on information design:} This work lies within the framework of information design. In \cite{bergemann2017information}, the authors suggested a unified framework for the information design problem. They consider a finite set of agents with a finite set of actions. The payoff of each agent depends on its own type and a state of the nature unknown to her. The designer knows the state and also the type associated with each agent. He will decide an information structure to satisfy his own objective. To provide a characterization of the set of equilibria achieved by controlling the information, the authors define the concept of obedience. Obedience is given by the set of "recommendations" provided by the designer that agents will follow (in a Bayesian setting). This specific set of strategies can be proved to be equal to the set of Bayesian correlated equilibria, which is characterized by linear algebraic constraints. Once this set is known, a Bayesian correlated equilibrium is chosen to maximize the objective of the designer and an extended information structure is suggested by the designer such that the desired Bayesian Nash equilibrium is equal to a Bayesian correlated equilibrium. This can be proven to be equivalent to a linear optimization problem.
 For a general overview of information design in the previously described set up, from the application point of view or from the algorithmic point of view, see the survey papers \cite{bergemann2017information,dughmi2017algorithmic,horner}. 
 This body of works has significant differences with our approach. In \cite{bergemann2017information}, the authors assume that the designer knows the state of nature and will design a probabilistic signal mechanism such that induces agents to behave according to the desire of the designer. A typical example is how to create an incentive for companies to invest in a risky investment by hiding the risky nature of the investment using information design.  
In our work, we design the information that is supplied to the society. As such, our model is not captured within the framework of information design suggested in \cite{bergemann2017information}. We also use a deterministic signal mechanism (such as quantization) instead of a probabilistic mechanism. Finally the framework suggested in \cite{bergemann2017information} is not adapted to a continuum of types/agents/actions. \\

\noindent \textbf{Social comparison and norm based approach: } Our model builds on \cite{benabou2006incentives}, in which the authors define a game theoretical model
 where the agents have to decide the intensity of their prosocial activity based on the following: their intrinsic motivation, a reward associated to the prosocial activity, a cost for performing the prosocial action and a desire to be perceived as an altruistic agent. 
 In this  setting, agents have private information about their intrinsic motivations but otherwise do not have full information concerning the intrinsic motivations of others. An information designer who, having full information about the actions, provides limited information about all agents' actions in order to aggregate prosocial action.
 Extension of this model has been studied in \cite{benabou2011laws,ali2016image}. In \cite{benabou2006incentives}, the authors examine crowding out effects when the service provider is increasing the reward for performing a prosocial activity. In \cite{benabou2011laws}, a similar model is studied, but agents have imperfect knowledge of the distribution of the intrinsic values. They then study an information design problem where the question is whether or not the designer should use "norms-based interventions" (in other words to reveal or not to reveal more information concerning the distribution of the intrinsic value). In \cite{ali2016image}, the authors try to understand the impact of making the action of an agent visible to the community. By doing so the designer can increase the effect of the reputational benefit. However, at the same time, the privacy of an agent is violated. Our work captures different aspects of the problem which are complementary to \cite{benabou2006incentives,benabou2011laws,ali2016image}. An important similarity between these works and ours is the fact that the reputation benefit is modeled as a signal extraction problem. The major differences are the following: 1) In \cite{benabou2006incentives,benabou2011laws,ali2016image}, the authors consider that the prosocial action of an agent is revealed to a proportion of the society. However when an action is revealed it is fully revealed as opposed to our model. The information design is on controlling the proportion of the society to whom information is revealed or on the amount of initial knowledge of the agent concerning the parameter of the game. In our work, only a quantized version of the feedback is revealed to the agents. 2) This new feedback requires a new mathematical proof of the existence of the mean field equilibrium. 3) We address the privacy dilemma from another perspective. Instead it is to provide partial privacy is not to limit the fraction of the network that gets to know the action of an agent. Our method is to provide a quantized version of an agent's action and thereby respect privacy to some degree.\\

\noindent \textbf{Other related works:} There is a wide range of applications of information design such as routing game \cite{das2017reducing,acemoglu2016informational}, queueing game \cite{sharma2017queuing,altman2013admission}, economic applications concerning persuasion \cite{kamenica2011bayesian,bergemann2016information}, and matching markets \cite{ostrovsky2010information}. These works are not using the framework of information design defined in \cite{bergemann2017information} or the one defined in our paper. But still, information (size of the queue, roads available) is hidden to improve the efficiency of a given system (delay in a queue, congestion in a city), and in that respect is related to our work.

\section{Discussion and extensions.}\label{sec: extension}

In this section, we will discuss the different extensions of our work. \\

\noindent \textbf{Convex cost functions.} Although all our results have been proved for a quadratic cost function, extension to non-decreasing convex cost functions is direct. The non-decreasing assumption is  common  in economics: more prosocial actions have higher cost. The convexity preserves monotonicity in $w_i$ of the best response. The proofs for existence will be similar to the quadratic case.\\

\noindent \textbf{Equilibrium selection and learning algorithm.} An interesting question would be to understand if there is a natural decentralized learning algorithm that converges to the equilibrium induced by the different feedback. For instance, we can imagine a mechanism that mimics the behavior of an agent in the society. She will learn her reputation (or rank) and will adapt her prosocial action over time depending on her current reputation level.\\

\noindent \textbf{Non-linear reputational benefit.}  In this paper we study linear reputation benefit in the sense that $\mathbb{E}[w_i\mid L(\cdot)]$ is the expectation of a linear function of $w_i$. Therefore, we did not observe effects when an agent who is already perceived as prosocial relaxes on the intensity of her prosocial action because it is costly to maintain a high reputation in the society. It would be interesting to extend this work to the case where the reputational benefit is equal to $\mathbb{E}[s(w_i)\mid L(\cdot)]$, where $s(\cdot)$ is a concave function of $w_i$. Type-B feedback will have similar results, but for type-A, a careful study of the ordinary differential equation that appears in the proof of Theorem \ref{prop:fixed point type c} will be needed.\\

\noindent \textbf{Common resource sharing.}  In the current model, the interaction between the agents is only captured through the reputational benefit. However, it could be that the lesser the level of prosocial actions of agents in the society, the higher of the cost of an agent's action. Indeed, consider a routing game on a network with three roads and with one of them being really cheap (Braess's Paradox). The norm maybe that taking the costly road is a prosocial action. The cost of an agent's action also depends on the fraction of the population that take the same road. In this new framework, we need to extend the classical framework of routing game by adding the reputational benefit to the cost function.  There will be a need to adapt the classical results of routing games, and the ones of this paper, to prove the existence of an equilibrium.\\

\noindent \textbf{Creation of a collective identity.} Our focus in this paper has been on how to use social comparison and reputational benefit for inducing prosocial behavior at the individual level. From a larger perspective, however, it would be interesting to explore ways to create and maintain a collective identity or a collective awareness. For example, a group of people in a neighborhood could be induced to form a team and meet team goals of reduced consumption via social comparison with other similar groups. Another example could be to provide the collective with a reduced consumption goal. Then one could indicate how much an individual's prosocial action has contributed towards the collective goal, such as how many kg of $CO_2$ emission has been saved. The challenge in such collectives is to ensure that interest is sustained over a sufficiently long duration. Methods that can help sustain such interest could be a topic of future research.\\

\noindent \textbf{Social network for immediate impact.} Education can help improve the ``starting point'', for example, shift the distribution $F$ of the propensity for prosocial action. Policies to improve this starting point may however not be effective in the short term. In contrast, our approach has been to make use of the social network to obtain a more immediate impact. Networks are profoundly changing the way people aggregate preferences, and our approach to make use of social comparison to induce prosocial behavior is in line with this. It may be interesting to see how to combine the long-term policy-based approaches with the social comparison approaches.

\section{Conclusions.}\label{sec: conclusion}

In this paper, our goal has been to show the importance of information design to improve prosocial behavior. We considered two types of feedback, one without privacy that revealed all actions to all agents in the network, and another that provided only quantized information about agents' actions. We extended the initial model of \cite{benabou2006incentives} to more complex information designs which required us to derive new proofs of existence, sufficient conditions for uniqueness, and characterization of the mean field equilibria. When the intrinsic motivation is drawn from the uniform distribution we obtained an explicit expression for the mean field equilibria. Moreover, we identified a setting where it was beneficial to be economical with the feedback information to improve the expected prosocial action. Finally, we suggested several possible extensions and future directions which we feel will be of use to researchers in the field. 

\section{Appendix: Game reformulation \textit{\`a la} Crawford and Sobel \cite{crawford1982strategic}.}\label{appendix: Game Reformulation}
The socio-economic setting behind the computation of the reputational benefit is the following. (1) The society observes for each agent  $i\in[0,1]$ a quantized feedback $L(a_i)$ of her action. (2) Using this collection of observations, the society estimates, for each agent $i$, the intrinsic motivation $E[w_i\mid\{L(\cdot)\}]$. (3) Agent $i$ cares about her reputation and derives a reputational benefit $E[w_i\mid\{L(\cdot)\}]$.

From a strict mathematical perspective, we have been ambiguous in not specifying the measure space over which the continuum of random variables $\{a_i\}_{i\in \mathcal{I}}$ are measurable. The $\sigma$-algebras over which the conditional measures of $w_i$ given $\{L(a_i)\}_i$ are defined has been also left unclear. This is the intuitive language adopted in some literature for ease of exposition without obscure mathematical formalism. 

 We now provide the necessary mathematical formalism. The game is composed of only two players, the agent (leader/sender) and the society (follower/receiver). Let $w$ be the type of the agent distributed according to $F(\cdot)$. The realization $w$ is known only to the agent, but $F(\cdot)$ is common knowledge. Let $a\in[0,+\infty)$ be the action of the agent. Let $L(a)$ be revealed to the society. Let $y\in[0,+\infty)$ be the action of the society.  The players' respective utilities are given by the following.
\begin{enumerate}
    \item The agent's utility: 
    \begin{equation}
        U(a,w,y)=aw-C(a)+\beta y,
    \end{equation}
    \item The society's utility: 
    \begin{equation}
        V(w,y)=-(w-y)^2.
    \end{equation}
\end{enumerate}
Ideally, the society wants to choose $y$ equal to $w$. However, the society does not know $w$ and only observes $L(a)$.

The strategy of the agent, denoted $a(w)$, can be relaxed to be a statistical experiment $\sigma_a: w \rightarrow[0,+\infty)$. We now suggest a framework, based on Crawford \& Sobel \cite{crawford1982strategic}, in which the results of this paper can be interpreted in a mathematically rigorous way. 

The statistical experiment $\sigma_a$ yields on an the associated statistical experiment $\sigma_L: w\rightarrow[0,+\infty)$). After observing $L(a)$, the society will update the prior belief $F$ to the a posterior belief, whose density is given by
\begin{equation}\label{eqn: update prior}
\mathbb{P}(w\mid L(a))=\frac{\sigma_L(L(a)\mid w)f(w)}{\int_0^{+\infty}\sigma_L(L(a)\mid v)f(v)\,dv}.
\end{equation}
From \eqref{eqn: update prior}, we can deduce that the society's best reply to $\sigma_a$ is given by:
\begin{eqnarray}\nonumber
y^*(\sigma_a)&=&\text{arg } \max_{y\in[0,+\infty)}-\int_0^{+\infty}\mathbb{P}(w\mid L(a))(w-y)^2\,dw,\\\nonumber
&=&\mathbb{E}[w\mid L(a)],
\end{eqnarray}
the minimum mean squared error estimate under the joint law coming from $F$, $\sigma_a$ for $(X,A) $ and the associated $F,\sigma_L$ for $(X,L)$.
The best reply to $y(\sigma_a)$ for the agent is given by $\sigma_a^*$ which is any measure supported on the set: 
\begin{equation}
\displaystyle\text{arg max}_{a\in[0,+\infty)}aw-C(a)+\beta y^*(\sigma_a).
\end{equation}
 The equilibrium we study in this paper is $(\sigma_a^*,y^*(\sigma_a^*))$. In the setting of Crawford and Sobel, the $L$ was under the control of the leading player. Our setting differs in that $L$ is under the control of the information designer.

\bibliographystyle{abbrv} 
\bibliography{references}



\end{document}